\documentclass[11pt]{article}

\usepackage{amsmath}
\usepackage{amsthm}
\usepackage{amssymb}
\usepackage{titlesec}
\usepackage[thinlines,thiklines]{easybmat}

\usepackage{graphicx}
\theoremstyle{definition}\newtheorem{Df}{Definition}
\theoremstyle{plain}\newtheorem{Th}{Theorem}
\theoremstyle{definition}\newtheorem{Rm}{Remark}
\theoremstyle{definition}
\theoremstyle{plain}
\theoremstyle{plain}\newtheorem{Co}[Th]{Corollary}
\theoremstyle{plain}\newtheorem{Lm}[Th]{Lemma}

\begin{document}
\title{ State succinctness of two-way finite automata with quantum and classical states \footnote{This work is supported in
part by the National Natural Science Foundation of China (Nos.
60873055, 61073054,61100001), the Natural Science Foundation of Guangdong
Province of China (No. 10251027501000004), the Fundamental
Research Funds for the Central Universities (Nos.
10lgzd12,11lgpy36), the Research Foundation for the Doctoral
Program of Higher School of Ministry of Education (Nos.
20100171110042, 20100171120051) of China,  the China Postdoctoral Science
Foundation project (Nos. 20090460808, 201003375), and the project
of  SQIG at IT, funded by FCT and EU FEDER projects projects QSec
PTDC/EIA/67661/2006, AMDSC UTAustin/MAT/0057/2008, NoE Euro-NF,
and IT Project QuantTel, FCT project PTDC/EEA-TEL/103402/2008
QuantPrivTel.}}

\author{Shenggen Zheng$^{1,3,}$\thanks{{\it  E-mail
address:} zhengshenggen@gmail.com (S. Zheng)},\hskip 2mm Daowen
Qiu$^{1,2,4,}$\thanks{Corresponding author. {\it E-mail address:}
issqdw@mail.sysu.edu.cn (D. Qiu)}, \hskip 2mm Jozef Gruska$^{3}$,
\hskip 2mm Lvzhou Li$^{1}$,
 \hskip 2mm Paulo Mateus$^{2}$
 \\
\small{{\it $^{1}$Department of
Computer Science, Sun Yat-sen University, Guangzhou 510006,
  China }}\\
\small {{\it $^{2}$ SQIG--Instituto de Telecomunica\c{c}\~{o}es, Departamento de Matem\'{a}tica,}}\\
\small {{\it  Instituto Superior T\'{e}cnico, TULisbon, Av. Rovisco Pais
1049-001, Lisbon, Portugal}}\\
\small{{\it $^{3}$ Faculty of Informatics, Masaryk University, Brno, Czech Republik}}\\
\small{{\it $^{4}$ The State Key Laboratory of Computer Science, Institute of Software,}}\\
\small{{ \it Chinese  Academy of Sciences, Beijing 100080, China}}
}

\date{ }
\maketitle \vskip 2mm \noindent {\bf Abstract}
\par
{\it Two-way quantum automata with quantum and classical states}
(2QCFA) were introduced by Ambainis and Watrous in 2002. In this
paper we study state succinctness of 2QCFA.
 For any $m\in {\mathbb{Z}}^+$ and any $\epsilon<1/2$, we show that: \begin{enumerate}
                                                   \item there is a promise problem $A^{eq}(m)$ which can be solved by a 2QCFA with one-sided error $\epsilon$ in a polynomial expected running time with a constant number (that depends neither on $m$ nor
        on $\varepsilon$) of quantum states and $\mathbf{O}(\log{\frac{1}{\epsilon})}$ classical states, whereas the sizes of the corresponding {\it deterministic finite automata} (DFA),  {\it two-way nondeterministic finite automata} (2NFA) and polynomial expected running time {\it two-way probabilistic finite automata} (2PFA) are at least $2m+2$, $\sqrt{\log{m}}$, and $\sqrt[3]{(\log m)/b}$, respectively;
                                                   \item there exists a language $L^{twin}(m)=\{wcw| w\in\{a,b\}^*\}$ over the alphabet $\Sigma=\{a,b,c\}$ which can be recognized by  a 2QCFA with one-sided error $\epsilon$ in an exponential expected running time with a constant number of quantum states and $\mathbf{O}(\log{\frac{1}{\epsilon})}$ classical states, whereas the sizes of the corresponding DFA, 2NFA and polynomial expected running time 2PFA are at least $2^m$, $\sqrt{m}$, and $\sqrt[3]{m/b}$, respectively;
                                                 \end{enumerate}
 where $b$ is a constant.

\par
\vskip 2mm {\sl Keywords:}   Computing models; Quantum finite automata; State complexity; Succinctness.
\vskip 2mm

\section{Introduction}

An important way to get a deeper insight into the power of various
quantum resources and features for information processing is to
explore the power of various quantum variations of the basic
models of classical automata. Of a special interest and importance
is to do that for various quantum variations of classical
 finite automata because quantum resources are not cheap and quantum
operations are not easy to implement. Attempts to find out how
much one can do with very little of quantum resources and
consequently with the
 simplest quantum variations of classical finite automata are therefore
of a particular interest. This paper is an attempt to contribute
to such line of research.

There are two basic approaches how to introduce quantum features
to classical models of finite automata. The first one is to
consider quantum variants of the classical {\it one-way
(deterministic) finite automata} (1FA or 1DFA) and the second one
is to consider quantum variants of the classical {\it two-way
finite automata} (2FA or 2DFA). Already the very first attempts to
introduce such models, by Moore and Crutchfields \cite{Moo97} and
Kondacs and Watrous \cite{Kon97} demonstrated that in spite of the
fact that in the classical case, 1FA and 2FA have the same
recognition power, this is not so for their quantum variations (in
case only unitary operations and projective measurements are
considered as quantum operations). Moreover, already the first
important model of {\it two-way quantum finite automata} (2QFA),
namely that introduced by Kondacs and Watrous, demonstrated that
very natural quantum variants of 2FA are much too powerful - they
can recognize even some {\it non-context free languages} and are
actually not really finite in a strong sense \cite{Kon97}. It
started to be therefore of interest to introduce and explore some
``less quantum" variations of 2FA and their power
\cite{Amb06,Amb98,Amb02,Bro99,Ber03,Li06,Li08,LiQiu09,Mer06,Pas00,Qiu09,QiuYu09,Yak10,Yak10-2,Yak11}.

A very natural ``hybrid" quantum variations of 2FA, namely, {\it
two-way quantum automata with quantum and classical states}
(2QCFA) were introduced by Ambainis and Watrous \cite{Amb02}.
Using this model they were able to show, in an elegant way, that
an addition of a single qubit to a classical model can enormously
increase the power of automata. A 2QCFA is essentially a classical
2FA augmented with a quantum memory of constant size (for states
in a fixed Hilbert space) that does not depend on the size of the
(classical) input. In spite of such a restriction, 2QCFA have been
shown to be more powerful than {\it two-way probabilistic finite
automata} (2PFA) \cite{Amb02}.

State complexity and succinctness results are an important
research area of classical automata theory, see \cite{Yu98}, with
a variety of applications. Once quantum versions of classical
automata were introduced  and explored, it started to be of large
interest to find out, also through succinctness results, a
relation between the power of classical and
 quantum automata models. This has turned out to be an area of surprising
 outcomes that again indicated that the relations between
 classical and corresponding quantum automata models are intriguing. For
 example, it has been shown, see \cite{Amb98,AmYa11,Amb09,AmbNay02,Le06}, that for some languages 1QFA
 require exponentially less states than classical 1FA, but for some other
 languages it can be in an opposite way.

Because of the simplicity, elegance and interesting properties of
the 2QCFA model, as well as its natural character, it seems to be
both useful and interesting  to explore state complexity and
succinctness results of 2QCFA and this we will do in this paper.

In the first part of this paper, 2QCFA are recalled formally and
some basic notations are given. Then we will prove state
succinctness result of 2QCFA on an infinite family of promise
problems. For any $m\in {\mathbb{Z}}^+$ let
$A^{eq}_{yes}(m)=\{w\in\{a,b\}^*|w=a^mb^m\}$ and
$A^{eq}_{no}(m)=\{w\in\{a,b\}^*| w\neq a^mb^m \ {\it and} \
|w|\geq m\}$. For any $\epsilon<1/2$ ($\epsilon$ is always a
nonnegative number in this paper),
 we will prove that the promise problem $A^{eq}(m)=(A^{eq}_{yes}(m), A^{eq}_{no}(m))$ can be solved by a 2QCFA with one-sided error $\epsilon$ in a polynomial expected running time with a constant number of quantum states and $\mathbf{O}(\log{\frac{1}{\epsilon})}$ (the base of logarithm is always $2$ in this paper) classical states , whereas sizes of the corresponding DFA, 2DFA and 2NFA are at least $2m+2$, $\sqrt{\log{m}}$ and $\sqrt{\log{m}}$, respectively. We also show that
for any $m\in {\mathbb{Z}}^+$, any 2PFA solves the promise problem
$A^{eq}(m)$ with an error probability $\epsilon<1/2$ and within
polynomial expected running time has least $\sqrt[3]{(\log m)/b}$
states, where $b>0$ is a constant.  Finally, we show a state
succinctness result of 2QCFA on an infinite family of languages.
For any $m\in {\mathbb{Z}}^+$ and any $\epsilon<1/2$, there exists
a 2QCFA that recognizes language $L^{twin}(m)=\{wcw|
w\in\{a,b\}^*\}$ over the alphabet $\Sigma=\{a,b,c\}$ with
one-sided error $\epsilon$ in an exponential expected running time
with a constant number of quantum states and
$\mathbf{O}(\log{\frac{1}{\epsilon})}$ classical states . We use
lower bound of {\it communication complexity} to prove that any
DFA recognizing language $L^{twin}(m)$ has at least $2^m$ states.
Next, we prove that the sizes of the corresponding 2DFA and 2NFA
to recognize $L^{twin}(m)$ are at least $\sqrt{m}$. We also show
that for any $m\in {\mathbb{Z}}^+$, any 2PFA recognizing
$L^{twin}(m)$ with an error probability $\epsilon<1/2$ and within
polynomial expected running time has least $\sqrt[3]{ m/b}$
states, where $b>0$ is a constant.

We now outline the remainder of this paper. Definition of 2QCFA
and some auxiliary lemmas are recalled in Section 2. In Section 3
we prove a state succinctness result of 2QCFA on an infinite
family of promise problems. Then we show a state succinctness
result of 2QCFA on an infinite family of languages in Section 4.
Finally, Section 5 contains a conclusion and some open problems.

\section{Preliminaries}
 In the first part of this section we formally recall the model of
2QCFA we will use.
 Concerning the basics
of {\it quantum computation} we refer the reader to
\cite{Gru99,Nie00}, and concerning the basic properties of
automata models,  we refer the reader to
\cite{Gru99,Gru00,Hop79,Paz71,QiuLi08,Qiu12}.

\subsection{2QCFA}
2QCFA were first introduced by Ambainis and Watrous \cite{Amb02},
and then studied by Qiu, Yakaryilmaz and etc.
\cite{Qiu08,Yak10,ZhgQiu11,ZhgQiu11-2}.
 Informally, we describe a 2QCFA as a 2DFA which has an access to a quantum
 memory of a constant size (dimension), upon which it performs
quantum unitary transformations or projective measurement. Given a
finite set of quantum states $Q$, we denote by $\mathcal{H}(Q)$
the Hilbert space spanned by $Q$. Let
$\mathcal{U}(\mathcal{H}(Q))$ and $\mathcal{O}(\mathcal{H}(Q))$
denote the sets of unitary operators and projective measurements
over $\mathcal{H}(Q)$, respectively.
\begin{Df}
A 2QCFA $\mathcal{A}$ is specified by a 9-tuple
\begin{equation}
\mathcal{A}=(Q,S,\Sigma,\Theta,\delta,q_{0},s_{0},S_{acc},S_{rej})
\end{equation}
where:

 \begin{enumerate}
\item $Q$ is a finite set of quantum states;

\item $S$ is a finite set of classical states;

\item $\Sigma$ is a finite set of input symbols; $\Sigma$ is then
extended to the tape symbols set $\Gamma=\Sigma\cup\{\
|\hspace{-1.5mm}c,\$\}$, where $\ |\hspace{-1.5mm}c\notin \Sigma $
is called the left end-marker and $\$\notin \Sigma$ is called the
right end-marker;

\item $q_{0}\in Q$ is the initial quantum state;

\item $s_{0}\in S$ is the initial classical state;

\item $S_{acc}\subset S$ and $S_{rej}\subset S$ satisfying
$S_{acc}\cap S_{rej}=\emptyset$ are the sets of classical
accepting and rejecting states, respectively.

\item $\Theta$ is the transition function of quantum states:
\begin{equation}
\Theta:S\setminus(S_{acc}\cup S_{rej})\times \Gamma \rightarrow
\mathcal{U}(\mathcal{H}(Q))\cup \mathcal{O}(\mathcal{H}(Q)).
\end{equation}
Thus,
 $\Theta(s,\gamma)$ is either a
unitary transformation or a projective measurement.

\item $\delta$ is the transition function of classical states.

\begin{enumerate}

  \item [a)]  If $\Theta(s, \gamma)\in \mathcal{U}(\mathcal{H}(Q))$, then \begin{equation}
\delta:S\setminus(S_{acc}\cup S_{rej})\times \Gamma \rightarrow
S\times \{-1, 0, 1\},
\end{equation} which is similar to the transition function for 2DFA, $\delta(s,\gamma)=(s',d)$ means
that when the classical state  $s\in S$ scanning $\gamma\in
\Gamma$ is changed to state $s'$, and the movement of the tape
head is determined by $d$ (moving right one cell if $d=1$, left if
$d=-1$, and being stationary if $d=0$).

\item [b)]If $\Theta(s, \gamma)\in \mathcal{O}(\mathcal{H}(Q))$,
then we assume that $\Theta(s,\gamma)$ is a projective measurement
with a set of possible eigenvalues $R=\{r_1, \cdots, r_n\}$ and
the projectors set $\{P(r_i):i=1, \cdots, n\}$, where $P(r_i)$
denotes the projector onto the eigenspace corresponding to $r_i$.
In such a case
\begin{equation}
\delta:S\setminus(S_{acc}\cup S_{rej})\times \Gamma \times
R\rightarrow S\times \{-1, 0, 1\},
\end{equation}
 where $\delta(s,\gamma)(r_{i})=(s',d)$ means
that when the projective measurement result is $r_{i}$, the
classical state  $s\in S$ is changed to $s'$, and the movement of
the tape head is determined by $d$.
\end{enumerate}
\end{enumerate}
\end{Df}

Given an input $w$, a 2QCFA
$\mathcal{A}=(Q,S,\Sigma,\Theta,\delta,q_{0},s_{0},S_{acc},S_{rej})$
proceeds as follows: at the beginning, the tape head is positioned
on the left end-marker $|\hspace{-1.5mm}c$, the quantum initial
state is $|q_{0}\rangle$, the classical initial state is $s_{0}$.
In the next steps if the current quantum state is $|\psi\rangle$,
the current classical state is $s\in S\setminus(S_{acc}\cup
S_{rej})$ and the current scanning symbol is $\sigma\in\Gamma$,
then the quantum state $|\psi\rangle$ and the classical state $s$
will be changed according to $\Theta(s,\sigma)$ as follows:

\begin{enumerate}

\item if $\Theta(s,\sigma)$ is a unitary operator $U$, then $U$ is
applied to the current quantum state $|\psi\rangle$ changing it
into $U|\psi\rangle$, and $\delta(s,\sigma)=(s',d)\in S\times
\{-1,0,1\}$ makes the current classical state $s$ to become $s'$,
together with the tape head moving in terms of $d$. In case $s'\in
S_{acc}$, the input is accepted, and in case $s'\in Q_{rej}$, the
input rejected;

\item if $\Theta(s,\sigma)$ is a projective measurement, then the
current quantum state $|\psi\rangle$ is changed to the quantum
state $P_{j}|\psi\rangle/ \|P_{j}|\psi\rangle\|$ with probability
$\|P_{j}|\psi\rangle\|^{2}$ in terms of the measurement, and in
this case, $\delta(s,\sigma)$ is a mapping from the set of all
possible results of the measurement to $S\times \{-1,0,1\}$. For
instance, for the result $r_j$ of the measurement, and
$\delta(s,\sigma)(r_j)=(s_{j},d)$, we have
\begin{enumerate}
\item if $s_{j}\in S\setminus (S_{acc}\cup S_{rej})$, new
classical state is $s_{j}$ and the head moves in the direction
$d$; \item

if $s_{j}\in S_{acc}$, the machine accepts the input and the
computation halts;

\item  and similarly, if $s_{j}\in S_{rej}$, the machine rejects
the input and the computation halts.
\end{enumerate}
It is seen that if the current all possible classical states are
in $S_{acc}\cup S_{rej}$, then the computation for the current
input string ends.
\end{enumerate}

 The computation will end if classical state
is in $S_{acc}\cup S_{rej}$. Therefore, similar to the definition
of accepting and rejecting probabilities for 2QFA \cite{Kon97},
the accepting and rejecting probabilities $Pr[\mathcal{A}\  {\it
accepts}\  w]$ and $Pr[\mathcal{A}\ {\it rejects}\  w]$ in
$\mathcal{A}$ for input $w$ are respectively the sums of all
accepting probabilities and all rejecting probabilities before the
end of the machine for computing input $w$.

Let $L\subset \Sigma^*$ and $\epsilon<1/2$. A 2QCFA $\mathcal{A}$
recognizes $L$ with one-sided error $\epsilon$ if

\begin{enumerate}
\item[1.] $\forall w\in L$, $Pr[\mathcal{A}\  {\it accepts}\
w]=1$, and \item[2.] $\forall w\notin L$, $Pr[\mathcal{A}\ {\it
rejects}\  w]\geq 1-\epsilon$.
\end{enumerate}

\subsection{Notations and auxiliary lemmas}
 In this subsection we review some additional notations related to 2QCFA \cite{Qiu08}.
For convenience, let $2QCFA_{\epsilon}$ denote the classes of all
languages recognized by 2QCFA with a given error probability
$\epsilon$ and $2QCFA_{\epsilon}(ptime)$ denote the classes of
languages recognized in polynomial expected time by 2QCFA with a
given error probability $\epsilon$. Moreover, let
$QS(\mathcal{A})$ and $CS(\mathcal{A})$ denote the numbers of
quantum states and classical states of a 2QCFA $\mathcal{A}$ and
let $T(\mathcal{A})$ denote the expected running time of 2QCFA
$\mathcal{A}$. For a string $w$, the length of $w$ is denoted by
$|w|$.

\begin{Lm}[\cite{Amb02}]\label{Leq}
For any $\epsilon<1/2$, there is a 2QCFA $\mathcal{A}(\epsilon)$
that accepts any $w \in L^{eq}=\{a^mb^m| m\in \mathbb{N} \}$ with
certainty, rejects any $w \notin L^{eq}$ with probability at least
$1-\epsilon$ and halts in expected running time
$\mathbf{O}(|w|^4)$, where $w$ is the input.
\end{Lm}

\begin{Rm}
According to the proof of Lemma \ref{Leq} in \cite{Amb02}, for the
above 2QCFA $\mathcal{A}(\epsilon)$ we further have
$QS(\mathcal{A}(\epsilon))=2$, $CS(\mathcal{A}(\epsilon))\in
\mathbf{O}(\log{\frac{1}{\epsilon})}$.
\end{Rm}

\begin{Lm}[\cite{Qiu08}]\label{Lm_2qcfa_inter}
If $L_1\in 2QCFA_{\epsilon_1}(2QCFA_{\epsilon_1}(ptime))$ and
$L_2\in 2QCFA_{\epsilon_2}(2QCFA_{\epsilon_2}(ptime))$, then
$L_1\cap L_2\in 2QCFA _{\epsilon}(2QCFA _{\epsilon}(ptime))$,
where $\epsilon=\epsilon_1+\epsilon_2-\epsilon_1\epsilon_2$.
\end{Lm}

\begin{Rm}
According to the proof of Lemma \ref{Lm_2qcfa_inter} in
\cite{Qiu08}, if 2QCFA $\mathcal{A}_1$ recognizes $L_1$ with
one-sided error $\epsilon_1$ (in polynomial expected time) and
2QCFA $\mathcal{A}_2$ recognizes $L_2$ with one-sided error
$\epsilon_2$ (in polynomial expected time), then there is a 2QCFA
$\mathcal{A}$ recognizes $L_1\cap L_2$ (in polynomial expected
time), where $QS(\mathcal{A})=QS(\mathcal{A}_1)+QS(\mathcal{A}_2)$
and
$CS(\mathcal{A})=CS(\mathcal{A}_1)+CS(\mathcal{A}_2)+QS(\mathcal{A}_1)$.
\end{Rm}

\begin{Lm}[\cite{She59,Var89}]\label{2DFAtoDFA}
Every $n$-state 2DFA can be simulated by a DFA with $(n+1)^{n+1}$
states.
\end{Lm}

\begin{Lm}[\cite{Bir93}]\label{2NFAtoDFA}
Every $n$-state 2NFA can be simulated by a DFA with
$2^{(n-1)^2+n}$ states.
\end{Lm}

\begin{Df}
Let language $L\subset \Sigma^*$ and $\epsilon<1/2$, then a 2PFA
$\mathcal{A}$ recognizes $L$ with error probability $\epsilon$ if
\begin{enumerate}
\item [{\it (1)}] $\forall w\in L$, $Pr[\mathcal{A}\  {\it
accepts}\  w]\geq 1-\epsilon$, and \item [{\it (2)}] $\forall
w\notin L$, $Pr[\mathcal{A}\ {\it rejects}\  w]\geq 1-\epsilon$.
\end{enumerate}

2PFA $\mathcal{A}$ recognizes $L$ if there is an $\epsilon<1/2$
such that $\mathcal{A}$ recognizes $L$ with error probability
$\epsilon$.
\end{Df}

\begin{Df}
Let $A,B\in \Sigma^*$  with $A\cap B=\varnothing$, then a 2PFA
$\mathcal{A}$ separates $A$ and $B$ \cite{DwS92} if there is some
$\epsilon<1/2$ such that
\begin{enumerate}
\item [{\it (1)}] $\forall w\in A$, $Pr[\mathcal{A}\  {\it
accepts}\  w]\geq 1-\epsilon$, and \item [{\it (2)}] $\forall w\in
B$, $Pr[\mathcal{A}\ {\it rejects}\  w]\geq 1-\epsilon$.
\end{enumerate}
\end{Df}

\begin{Lm}[\cite{DwS90}]\label{2PFAtoDFA}
For every $\epsilon<1/2$, $a>0$ and $d>0$, there exists a constant
$b>0$ such that, for any $c$, if $L$ is recognized by a $c$-state
2PFA with an error probability $\epsilon$ and within time $an^d$,
then $L$ is recognized by some DFA with at most $c^{bc^2}$ states,
where $n=|w|$ is the length of the input.
\end{Lm}

\begin{Lm}[\cite{DwS92}]\label{not-in-2PFA}
Let $A, B\subseteq \Sigma^*$ with $A\cap B=\varnothing$. Suppose
there is an infinite set $I$ of positive integers and, for each
$m\in I$, a set $W_m\subseteq\Sigma^*$ such that
\begin{enumerate}
  \item [{\it (1)}] $|w|\leq m$ for all $w\in W_m$,
  \item [{\it (2)}] for every integer $k$, there is an $m_k$ such that $|W_m|\geq m^k$ for all $m\in I$ with $m\geq m_k$, and
  \item [{\it (3)}] for every $m\in I$ and every $w, w'\in W_m$ with $w\neq w'$, there are words $u, v\in \Sigma^*$ such that either $uwv\in A$ and $uw'v\in B$ or $uwv\in B$ and $uw'v\in A$.
\end{enumerate}
Then no 2PFA separates $A$ and $B$.
\end{Lm}

We recall some basic notations of {\it communication complexity},
and we refer the reader to \cite{KusNis97,KusNis97b,Yao79} for
more details. It deals with the situation where there are only two
communicating  parties and it deals with very simple tasks of
{computing two argument functions where one argument is known to
one party and the other argument is known to the other party. It
completely ignores the computational resources needed by the
parties and it focuses solely on the amount of communication
exchanged between the parties.

Let $X, Y, Z$ be arbitrary finite sets. We consider a two-argument
function $f: X\times Y\rightarrow Z$ and two communicating
parties,  Alice is given an input $x\in X$ and Bob is given an
input $y\in Y$. They wish to compute $f(x,y)$.

The computation of the value $f(x,y)$ is done using a
communication protocol. During the execution of the protocol, the
two parties alternate roles in sending messages. Each of these
messages is a string of bits. The protocol, based on the
communication so far, specifies whether the execution terminated
(in which case it also specifies what is the output). If the
execution has not terminated, the protocol specifies what message
the sender (Alice or Bob) should send next, as a function of its
input and of the communication so far. A communication protocol
${\cal P}$ computes the function $f$, if for every input pair
$(x,y)\in A\times B$ the protocol terminates with the value
$f(x,y)$ as its output.

We define the {deterministic communication complexity} of ${\cal
P}$ as the worst case number of bits exchanged by the protocol.
The {deterministic communication complexity} of a function $f$ is
the communication complexity of the best protocol that computes
$f$, denoted by $D(f)$.

\begin{Lm}[\cite{KusNis97}]\label{ComC-EQ}
If Alice and Bob each holds an $n$ length string, $x,y\in
\{a,b\}^n$ and the equality function, $EQ(x,y)$, is defined to be
$1$ if $x=y$ and $0$ otherwise, then
\begin{equation}
  D(EQ)=n+1.
\end{equation}
\end{Lm}

\section{State succinctness of 2QCFA on promise problems}
In this section, we will give an infinite family of promise
problems\footnote{A promise problem $A = (A_{yes}, A_{no})$ is
solved by a 2QCFA $\mathcal{A}$ with one-sided error
$\epsilon<1/2$ if {\it (1)} $\forall w\in A_{yes}$,
$Pr[\mathcal{A}\  {\it accepts}\  w]=1$, and {\it (2)} $\forall
w\in A_{no}$, $Pr[\mathcal{A}\ {\it rejects}\  w]\geq 1-\epsilon$.
A promise problem $A = (A_{yes}, A_{no})$ is solved by a 2PFA
$\mathcal{A}$ with error probability $\epsilon<1/2$ if {\it (1)}
$\forall w\in A_{yes}$, $Pr[\mathcal{A}\  {\it accepts}\  w]\geq
1-\epsilon$, and {\it (2)} $\forall w\in A_{no}$, $Pr[\mathcal{A}\
{\it rejects}\  w]\geq 1-\epsilon$. A promise problem $A =
(A_{yes}, A_{no})$ is solved by a DFA (2DFA, 2NFA) $\mathcal{A}$
if {\it (1)} $\forall w\in A_{yes}$, $\mathcal{A}$ accepts $w$ and
{\it (2)} $\forall w\in A_{no}$, $\mathcal{A}$ rejects $w$. }
which can be solved by 2QCFA with one-sided error $\epsilon$ in a
polynomial expected running time with a constant number of quantum
states and  $\mathbf{O}(\log{\frac{1}{\epsilon})}$ classical
states.

A {\it promise problem} is a pair $A = (A_{yes}, A_{no})$, where
$A_{yes}$, $A_{no}\subset \Sigma^*$ are disjoint sets of strings
\cite{Wat09}. (Languages may be viewed as promise problems that
obey the additional constraint $A_{yes}\cup A_{no}=\Sigma^*$.) For
an alphabet $\Sigma=\{a,b\}$ and any $m\in {\mathbb{Z}}^+$, let
$A^{eq}_{yes}(m)=\{a^mb^m\}$ and $A^{eq}_{no}(m)=\{w\in\{a,b\}^*|
w\neq a^mb^m \ {\it and} \ |w|\geq m\}$. For any $\epsilon<1/2$,
we will prove that promise problems $A^{eq}(m)=(A^{eq}_{yes}(m),
A^{eq}_{no}(m))$ can be solved by a 2QCFA with one-sided error
$\epsilon$ in a polynomial expected running time with a constant
number of quantum states and
$\mathbf{O}(\log{\frac{1}{\epsilon})}$ classical states, whereas
the sizes of the corresponding DFA, 2DFA and 2PFA grow without a
bound.

In order to prove that the promise problem $A^{eq}(m)$ can be
solved by 2QCFA, we first prove that a simpler promise problem can
be solved by 2QCFA.

For an alphabet $\Sigma$ and an $m\in {\mathbb{Z}}^+$, let
$A_{yes}(m)=\{w\in \Sigma^* \mid |w|=m\}$ and $A_{no}(m)=\{w\in
\Sigma^*\mid |w|\neq m\ {\it and }\ |w|\geq m/2\}$. For any
$\epsilon<1/2$, we will prove that there is a 2QCFA that can solve
promise problem $A(m)=(A_{yes}(m), A_{no}(m))$ with one-sided
error $\epsilon$ in a polynomial expected running time with a
constant number of quantum states and
$\mathbf{O}(\log{\frac{1}{\epsilon})}$ classical states. The
language $L(m)=\{w\in \Sigma^*\mid |w|=m\}$ was showed to be
recognized by a 7-state {\it one way quantum finite automata with
restart} ($1QFA^\circlearrowleft$) with one-sided error $\epsilon$
in an exponential expected time by Yakaryilmaz and Cem Say
\cite{Yak10}. In the same paper, they mentioned that
$1QFA^\circlearrowleft$ can be simulated by 2QCFA easily. In
following theorem we will prove in details that the promise
problem $A(m)$ can be solved by a 2QCFA with one-sided error
$\epsilon$ in a polynomial expected time.

\begin{Th}\label{Am}
For any $m\in {\mathbb{Z}}^+$ and any $\epsilon<1/2$, there exists
a 2QCFA $\mathcal{A}(m,\epsilon)$ which accepts any $w\in
A_{yes}(m)$ with certainty, and rejects any $w\in A_{no}(m)$ with
probability at least $1-\epsilon$, where
$QS(\mathcal{A}(m,\epsilon))$ is a constant and
$CS(\mathcal{A}(m,\epsilon))\in
\mathbf{O}(\log{\frac{1}{\epsilon})}$. Furthermore, we have
$T(\mathcal{A}(m,\epsilon))\in\mathbf{O}(|w|^4)$, where $w$ is the
input.
\end{Th}
\begin{proof}
The main idea is as follows: we consider a 2QCFA
$\mathcal{A}(m,\epsilon)$ with 2 quantum states $|q_0\rangle$ and
$|q_1\rangle$. $\mathcal{A}(m,\epsilon)$ starts with the quantum
state $|q_0\rangle$. When $\mathcal{A}(m,\epsilon)$ reads the left
end-marker$\ |\hspace{-1.5mm}c$, the state is rotated by angle
$\sqrt{2}m\pi$ and every time when $\mathcal{A}(m,\epsilon)$ reads
a symbol $\sigma\in\Sigma^*$ , the state is rotated by angle
$-\alpha=-\sqrt 2\pi$ (notice that $\sqrt{2}m\pi=m\alpha$). When
the right end-marker $\$$ is reached, $\mathcal{A}(m,\epsilon)$
measures the quantum state. If it is $|q_1\rangle$, the input
string $w$ is rejected. Otherwise, the process is repeated.

 We now complete the description of $\mathcal{A}(m,\epsilon)$ as sketched in Figure \ref{t1}. The states of the automaton will be over the orthogonal base
$\{|q_0\rangle, |q_1\rangle\}$ and will use the following two unitary transformations\\

 \begin{tabular}{|l|l|}
  \hline
  $U_{|\hspace{-1.1mm}c}|q_0\rangle=\cos m\alpha |q_0\rangle+\sin m\alpha|q_1\rangle$ &  $U_{-\alpha}|q_0\rangle=\cos \alpha |q_0\rangle-\sin \alpha|q_1\rangle$
   \\

  $U_{|\hspace{-1.1mm}c}|q_1\rangle=-\sin m\alpha |q_0\rangle+\cos m\alpha|q_1\rangle$ & $U_{-\alpha}|q_0\rangle=\sin \alpha |q_0\rangle+\cos \alpha|q_1\rangle$ \\

  \hline
\end{tabular}\\

\begin{figure}
  \centering
  \begin{tabular}{|l|}
  \hline
Repeat the following ad infinitum:\\
1. Set the quantum state to $|q_0\rangle$, read the left end-marker $\ |\hspace{-1.5mm}c$, and perform $U_{|\hspace{-1mm}c}$ on $|q_0\rangle$. \\
2. Until the scanned symbol is the right end-marker $\$$, do the following:\\
\ \ (2.1). Perform $U_{-\alpha}$ on the current quantum state ($U_{\alpha}$ is defined in the proof of Theorem \ref{Am}).\\
\ \ (2.2). Move the tape head one square to the right.\\
3. Measure the quantum state. If the result is not $|q_0\rangle$, reject.\\
4. Repeat the following subroutine two times:\\
\ \ (4.1).Move the tape head to the first input symbol.\\
\ \ (4.2).Move the tape head one square to the right.\\
\ \ (4.3).While the currently scanned symbol is not$\ |\hspace{-1.5mm}c$ or $\$$, do the following:\\
\ \ \ \ \ \ Simulate a coin flip. If the result is ``head", move right. Otherwise, move left.\\
5. If both times the process ends at the right end-marker $\$$, do:\\
\ \ \ \ Simulate $k$ coin-flips and if all outcomes are ``heads", accept.\\
  \hline
\end{tabular}

  \centering\caption{ Description of the behaviour of 2QCFA $\mathcal{A}(m,\epsilon)$. The choice of $k$ will depend on $\epsilon$.}\label{t1}
\end{figure}

\begin{Lm}\label{m=n}
If the input $w\in A_{yes}(m)$, then the quantum state of
$\mathcal{A}(m,\epsilon)$ will evolve with certainty into
$|q_0\rangle$ after the loop \textbf{2}.
\end{Lm}
\begin{proof}
If $w\in A_{yes}(m)$, then $|w|=m$.   The quantum state after the
loop \textbf{2} can be described as follows:
\begin{equation}
|q\rangle=(U_{-\alpha})^m U_{|\hspace{-1.1mm}c}|q_0\rangle=\left(
  \begin{array}{cc}
    \cos \alpha  & \sin \alpha \\
    -\sin\alpha  & \cos \alpha \\
  \end{array}
\right)^m\left(
  \begin{array}{cc}
    \cos m\alpha  & -\sin m\alpha \\
    \sin m\alpha  & \cos  m\alpha \\
  \end{array}
\right)|q_0\rangle
\end{equation}
\begin{equation}
=\left(
  \begin{array}{cc}
    \cos m\alpha  & \sin m\alpha \\
    -\sin m\alpha  & \cos m\alpha \\
  \end{array}
\right)\left(
  \begin{array}{cc}
    \cos m\alpha  & -\sin m\alpha \\
    \sin m\alpha  & \cos  m\alpha \\
  \end{array}
\right)|q_0\rangle=\left(
  \begin{array}{cc}
    1  & 0 \\
    0  & 1 \\
  \end{array}
\right)|q_0\rangle=|q_0\rangle.
\end{equation}
\end{proof}

\begin{Lm}\label{m!=n}
If $w\in A_{no}(m)$, $|w|=n$, then $\mathcal{A}(m,\epsilon)$
rejects $w$ after the step \textbf{3} with a probability at least
$1/(2(m-n)^2+1)$.
\end{Lm}
\begin{proof}
Starting with the state $|q_0\rangle$, $\mathcal{A}(m,\epsilon)$
changes its quantum state to  $|q\rangle=(U_{-\alpha})^n
U_{|\hspace{-1.1mm}c}|q_0\rangle$ after the loop \textbf{2},
 the quantum state can be described as follows:
\begin{equation}
|q\rangle=(U_{-\alpha})^n U_{|\hspace{-1.1mm}c}|q_0\rangle=\left(
  \begin{array}{cc}
    \cos \alpha  & \sin \alpha \\
    -\sin\alpha  & \cos \alpha \\
  \end{array}
\right)^n\left(
  \begin{array}{cc}
    \cos m\alpha  & -\sin m\alpha \\
    \sin m\alpha  & \cos  m\alpha \\
  \end{array}
\right)|q_0\rangle
\end{equation}
\begin{equation}
=\left(
  \begin{array}{cc}
    \cos n\alpha  & \sin n\alpha \\
    -\sin n\alpha  & \cos n\alpha \\
  \end{array}
\right)\left(
  \begin{array}{cc}
    \cos m\alpha  & -\sin m\alpha \\
    \sin m\alpha  & \cos  m\alpha \\
  \end{array}
\right)|q_0\rangle
\end{equation}
\begin{equation}
=\left(
  \begin{array}{cc}
    \cos (m+n)\alpha  &\sin (m-n)\alpha \\
    \sin (m-n)\alpha  & \cos  (m+n)\alpha \\
  \end{array}
\right)|q_0\rangle=\cos((m-n)\alpha)|q_0\rangle+\sin((m-n)\alpha)|q_1\rangle.
\end{equation}

The probability of observing  $|q_1\rangle$ is
$\sin^2(\sqrt{2}(m-n)\pi)$ in the step \textbf{3}. Without loss of
generality, we assume that $m-n>0$. Let $l$ be the closest integer
to $\sqrt{2}(m-n)$. If $\sqrt{2}(m-n)>l$, then $2(m-n)^2>l^2$. So
we get $2(m-n)^2-1\geq l^2$ and $l\leq \sqrt{2(m-n)^2-1}$. We have
\begin{equation}
\sqrt{2}(m-n)-l\geq \sqrt{2}(m-n)-\sqrt{2(m-n)^2-1}
\end{equation}

\begin{equation}
=\frac{(\sqrt{2}(m-n)-\sqrt{2(m-n)^2-1})(\sqrt{2}(m-n)+\sqrt{2(m-n)^2-1})}{\sqrt{2}(m-n)+\sqrt{2(m-n)^2-1}}
\end{equation}
\begin{equation}
=\frac{1}{\sqrt{2}(m-n)+\sqrt{2(m-n)^2-1}}>\frac{1}{2\sqrt{2}(m-n)}.
\end{equation}

Because $l$ is the closest integer to $\sqrt{2}(m-n)$, we have
$0<\sqrt{2}(m-n)-l<1/2$. Let $f(x)=sin(x\pi)-2x$. We have
 $f''(x)=-\pi^2\sin(x\pi)\leq 0$ when $x\in
[0,1/2]$. That is to say, $f(x)$ is concave in $[0,1/2]$, and we
have $f(0)=f(1/2)=0$. So for any $x\in[0,1/2]$, it holds that
$f(x)\geq 0$, that is, $\sin(x\pi)\geq 2x$. Therefore, we have
\begin{equation}
\sin^2(\sqrt{2}(m-n)\pi)=\sin^2((\sqrt{2}(m-n)-l)\pi)
\end{equation}
\begin{equation}
\geq (2(\sqrt{2}(m-n)-l))^2 =4(\sqrt{2}(m-n)-l)^2
\end{equation}
\begin{equation}
> 4(\frac{1}{2\sqrt{2}(m-n)})^2=\frac{1}{2(m-n)^2}>\frac{1}{2(m-n)^2+1}.
\end{equation}

If $\sqrt{2}(m-n)<l$, then $2(m-n)^2<l^2$. So we get
$2(m-n)^2+1\leq l^2$ and $l\geq \sqrt{2(m-n)^2+1}$. We have
\begin{equation}
\sqrt{2}(m-n)-l\leq \sqrt{2}(m-n)-\sqrt{2(m-n)^2+1}
\end{equation}
\begin{equation}
=\frac{(\sqrt{2}(m-n)-\sqrt{2(m-n)^2+1})(\sqrt{2}(m-n)+\sqrt{2(m-n)^2+1})}{\sqrt{2}(m-n)+\sqrt{2(m-n)^2+1}}
\end{equation}
\begin{equation}
=\frac{-1}{\sqrt{2}(m-n)+\sqrt{2(m-n)^2+1}}<\frac{-1}{2\sqrt{2(m-n)^2+1}}.
\end{equation}

It follows that
\begin{equation}
l-\sqrt{2}(m-n)>\frac{1}{2\sqrt{2(m-n)^2+1}}.
\end{equation}
Because $l$ is the closest integer to $\sqrt{2}(m-n)$, we have
$0<l-\sqrt{2}(m-n)<1/2$. Therefore, we have
\begin{equation}
\sin^2(\sqrt{2}(m-n)\pi)=\sin^2((\sqrt{2}(m-n)-l)\pi)
\end{equation}
\begin{equation}
=\sin^2((l-\sqrt{2}(m-n))\pi)\geq (2(l-\sqrt{2}(m-n)))^2
\end{equation}
\begin{equation}
=4(l-\sqrt{2}(m-n))^2>
4(\frac{1}{2\sqrt{2(m-n)^2+1}})^2=\frac{1}{2(m-n)^2+1}.
\end{equation}
So the lemma has been proved.
\end{proof}

Simulation of a coin flip in the steps \textbf{4} and \textbf{5}
is a necessary component in the above algorithm. We will show that
coin-flips can be simulated by a 2QCFA using two quantum states
$|q_0\rangle$ and $|q_1\rangle$.

\begin{Lm}
A coin flip in the algorithm can be simulated by a 2QCFA
$\mathcal{A}(m,\epsilon)$ using two quantum states $|q_0\rangle$
and $|q_1\rangle$.
\end{Lm}
\begin{proof}
 Let us consider a projective measurement $M=\{P_0,P_1\}$ defined by
\begin{equation}
 P_0=|q_0\rangle\langle q_0|, P_1=|q_1\rangle\langle q_1|,
\end {equation}
whose classical outcomes will be denoted by 0 and 1, representing
the ``tail'' and ``head''  of a coin flip, respectively. Hadamard
unitary operator
\begin{equation}
  H=\left(%
\begin{array}{cc}
  \frac{1}{\sqrt{2}} &  \frac{1}{\sqrt{2}} \\
   \frac{1}{\sqrt{2}} &  -\frac{1}{\sqrt{2}} \\
\end{array}%
 \right).
\end {equation}
Hadamard operator changes basis states
\begin{equation}
|q_0\rangle\rightarrow|\psi\rangle=\frac{1}{\sqrt{2}}(|q_0\rangle+|q_1\rangle),\
\
|q_1\rangle\rightarrow|\phi\rangle=\frac{1}{\sqrt{2}}(|q_0\rangle-|q_1\rangle).
\end{equation}

Suppose now that the machine starts with the state $|q_0\rangle$,
changes its quantum state by $H$, and then measures the quantum
state with $M$. Then  we will get the result  0 or 1 with
probability
 $\frac{1}{2}$. This is similar to a coin flip
 process.
\end{proof}

\begin{Lm}\cite{Amb02}
If the length of the input string is $n$, then every execution of
the loops \textbf{4} and \textbf{5} leads to the acceptance with a
probability $1/2^k(n+1)^2$.
\end{Lm}
\begin{proof}
The loop \textbf{4} performs two times of random walk starting at
location 1 and ending at location 0 (the left end-marker\ \
$|\hspace{-1.5mm}c$) or at location $n+1$ (the right end-marker
$\$$). It is known from probability theory that the probability of
reaching the location $n+1$ is $1/(n+1)$ (see Chapter14.2 in
\cite{Fel67}). Repeating it twice and flipping $k$ coins, we  get
the probability $1/2^k(n+1)^2$.
\end{proof}

If we take $k=1+\lceil\log{\frac{1}{\epsilon}}\rceil$, then
$\epsilon\geq 1/2^{k-1}$.  Assume also that $|w|=n$. If $w\in
A_{yes}(m)$, the loop $\mathbf{2}$ always changes the quantum
state $|q_0\rangle$ to $|q_0\rangle$, and
$\mathcal{A}(m,\epsilon)$ never rejects after the measurement in
the step $\mathbf{3}$. After the loops $\mathbf{4}$ and
$\mathbf{5}$, the probability of $\mathcal{A}(m,\epsilon)$
accepting $w$ is $1/2^k(n+1)^2$. Repeating the loops $\mathbf{4}$
and $\mathbf{5}$ for $cn^2$ times, the accepting probability is
\begin{equation}
Pr[\mathcal{A}(m,\epsilon)\  accepts\  w]
=1-(1-\frac{1}{2^k(n+1)^2})^{cn^2},
\end{equation}
 and this can be made
arbitrarily close to 1 by selecting the constant $c$
appropriately.

Otherwise, if $|w|\in A_{no}(m)$, $\mathcal{A}(m,\epsilon)$
rejects the input after the steps $\mathbf{2}$ and $\mathbf{3}$
with probability
\begin{equation}
P_{r} >\frac{1}{2(m-n)^2+1}
\end{equation}
 according to Lemma \ref{m!=n}. $\mathcal{A}(m,\epsilon)$ accepts the input after the loops $\mathbf{4}$ and $\mathbf{5}$ with probability
\begin{equation}
P_{a}=1/2^k(n+1)^2\leq \epsilon/2(n+1)^2.
\end{equation}
If we repeat the whole algorithm indefinitely, the probability of
$\mathcal{A}(m,\epsilon)$ rejecting input $w$ is
\begin{equation}
Pr[\mathcal{A}(m,\epsilon)\  rejects\  w] =\sum_{i\geq
0}(1-P_a)^i(1-P_r)^iP_r
\end{equation}
\begin{equation}
=\frac{P_r}{P_a+P_r-P_aP_r}>\frac{P_r}{P_a+P_r}
\end{equation}
\begin{equation}
>\frac{1/(2(n-m)^2+1)}{\epsilon/2(n+1)^2+1/(2(n-m)^2+1)}
\end{equation}
\begin{equation}
=\frac{(n+1^2)/(2(n-m)^2+1)}{\epsilon/2+(n+1)^2/(2(n-m)^2+1)}
\end{equation}
Let
$f(x)=\frac{x}{\epsilon/2+x}=1-\frac{\epsilon}{(\epsilon+2x)}$,
then $f(x)$ is monotonous increasing in $(0, +\infty)$. By
assumption, we have $n=|w|\geq m/2$. So we have
$(n+1^2)/(2(n-m)^2+1)>1/2$. Therefore, we have
\begin{equation}
>\frac{1/2}{1/2+\epsilon/2}=\frac{1}{1+\epsilon}>1-\epsilon.
\end{equation}

If we assume the input is $w$, then the step $\mathbf{1}$ takes
$\mathbf{O}(1)$ time, the loop $\mathbf{2}$ and the step
$\mathbf{3}$ take $\mathbf{O}(|w|)$ time, and the loops
$\mathbf{4}$ and $\mathbf{5}$ take $\mathbf{O}(|w|^2)$ time. The
expected number of repeating the algorithm
 is $\mathbf{O}(|w|^2)$. Hence, the expected running
time of $\mathcal{A}(m,\epsilon)$ is $\mathbf{O}(|w|^4)$.
Obviously, $QS(\mathcal{A}(m,\epsilon))=2$.  We just need
$\mathbf{O}(k)$ classical states to simulate $k$ coin-flips and
calculate the outcomes, therefore $CS(\mathcal{A}(m,\epsilon))\in
\mathbf{O}(\log{\frac{1}{\epsilon})}$.
\end{proof}

\begin{Th}\label{meq}
For any $m\in {\mathbb{Z}}^+$ and any $\epsilon<1/2$, there exists
a 2QCFA $\mathcal{A}(m,\epsilon)$ which accepts any $w\in
A^{eq}_{yes}(m)$ with certainty, and rejects any $w\in
A^{eq}_{no}(m)$ with probability at least $1-\epsilon$, where
$QS(\mathcal{A}(m,\epsilon))$ is a constant and
$CS(\mathcal{A}(m,\epsilon))\in
\mathbf{O}(\log{\frac{1}{\epsilon})}$. Furthermore, we have
$T(\mathcal{A}(m,\epsilon))\in\mathbf{O}(|w|^4)$ where $w$ is the
input.
\end{Th}

\begin{proof}
Let the alphabet $\Sigma=\{a,b\}$. Obviously,
$A^{eq}(m)=L^{eq}\cap A(2m)$. According to Lemma \ref{Leq}, for
any $\epsilon_1>0$, there is a 2QCFA $\mathcal{A}_1(\epsilon_1)$
recognizes $L^{eq}$ with one-sided error $\epsilon_1$, and
$QS(\mathcal{A}_1(\epsilon_1))=2$,
$CS(\mathcal{A}_1(\epsilon_1))\in
\mathbf{O}(\log{\frac{1}{\epsilon_1})}$ and
$T(\mathcal{A}_1(\epsilon_1))\in\mathbf{O}(|w|^4)$. According to
Theorem \ref{Am}, for any $\epsilon_2>0$, there is a 2QCFA
$\mathcal{A}_2(m,\epsilon_2)$ that solves the promise problem
$A(2m)$ with one-sided error $\epsilon_2$, and
$QS(\mathcal{A}_2(m,\epsilon_2))=2$,
$CS(\mathcal{A}_2(m,\epsilon_2))\in
\mathbf{O}(\log{\frac{1}{\epsilon_2})}$ and
$T(\mathcal{A}_2(m,\epsilon_2))\in\mathbf{O}(|w|^4)$. For any
$\epsilon<1/2$, let $\epsilon_1=\epsilon/2$ and
$\epsilon_2=\epsilon/2$. According to Lemma \ref{Lm_2qcfa_inter},
there is a 2QCFA $\mathcal{A}(m,\epsilon)$ solves the promise
problem $L^{eq}\cap A(2m)$ with a one-sided error $\epsilon$,
where $QS(\mathcal{A}(m,\epsilon))=QS(\mathcal{A}_1(\epsilon_1))+
QS(\mathcal{A}_2(m,\epsilon_2))=4$,
$CS(\mathcal{A}(m,\epsilon))=CS(\mathcal{A}_1(\epsilon_1))+CS(\mathcal{A}_2(m,\epsilon_2))+QS(\mathcal{A}_1(\epsilon_1))\in
\mathbf{O}(\log{\frac{1}{\epsilon})}$ and
$T(\mathcal{A}(m,\epsilon))=T(\mathcal{A}_1(\epsilon_1))+T(\mathcal{A}_2(m,\epsilon_2))\in\mathbf{O}(|w|^4)$.
Hence, the theorem has been proved.
\end{proof}

\begin{Rm}
Actually, $L_1$ and $L_2$ must be  languages in Lemma
\ref{Lm_2qcfa_inter}. But in Theorem \ref{meq}, we used a promise
problem $A(2m)$. It is easy to show that Lemma
\ref{Lm_2qcfa_inter} still holds for promise problem $A(2m)$ and
language $L^{eq}$. We used Lemma \ref{Lm_2qcfa_inter} to prove
Theorem \ref{meq} in this section. However, we can prove Theorem
\ref{meq} directly.
\end{Rm}

Obviously, there exists a DFA depicted in Figure \ref{f1} that
solves the promise problem $A^{eq}(m)$ with 2m+2 states.
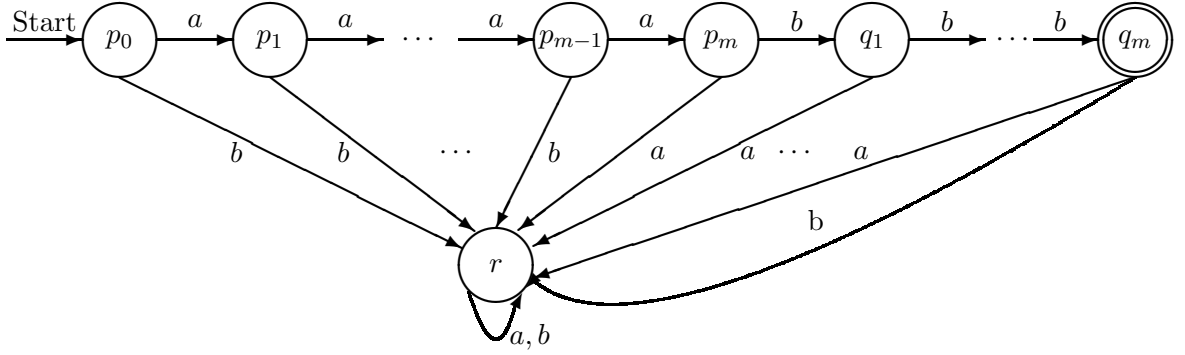
\begin{figure}
  \centering

  \setlength{\unitlength}{1cm}
\begin{picture}(30,4)\thicklines
\put(2,4){\circle{1}\makebox(0,0){$p_{0}$}}
\put(4,4){\circle{1}\makebox(0,0){$p_{1}$}}
\put(6,4){\makebox(0,0){$\cdots$}}
\put(8,4){\circle{1}\makebox(0,0){$p_{m-1}$}}

\put(10,4){\circle{1}\makebox(0,0){$p_m$}}
\put(12,4){\circle{1}\makebox(0,0){$q_1$}}

\put(13.9,4){\makebox(0,0){$\cdots$}}

\put(15.5,4){\circle{1}\circle{0.9}\makebox(0,0){$q_m$}}

\put(0.5,4){\vector(1,0){1}\makebox(-1,0.5){Start}}
\put(2.5,4){\vector(1,0){1}\makebox(-1,0.5){$a$}}
\put(4.5,4){\vector(1,0){1}\makebox(-1,0.5){$a$}}
\put(6.5,4){\vector(1,0){1}\makebox(-1,0.5){$a$}}
\put(8.5,4){\vector(1,0){1}\makebox(-1,0.5){$a$}}
\put(10.5,4){\vector(1,0){1}\makebox(-1,0.5){$b$}}
\put(12.5,4){\vector(1,0){1}\makebox(-1,0.5){$b$}}
\put(14.15,4){\vector(1,0){0.85}\makebox(-1,0.5){$b$}}

\put(7,1){\circle{1}\makebox(0,0){$r$}} \qbezier(6.64,0.64)
(7,-0.5) (7.25,0.44)
\put(7.25,0.44){\vector(1,2){0.1}\makebox(0.2,-0.8){$a,b$}}

\put(2,3.5){\vector(2,-1){4.55}\makebox(-6,-2){$b$}}
\put(3,3.5){\makebox(7,-2){$\cdots$}}
\put(4,3.5){\vector(4,-3){2.72}\makebox(-3.5,-2){$b$}}
\put(8,3.5){\vector(-1,-2){1}\makebox(1.2,-2){$b$}}
\put(10,3.5){\vector(-4,-3){2.7}\makebox(3,-2){$a$}}

\put(12,3.5){\vector(-2,-1){4.5}\makebox(5,-2){$a$}}
\put(13,3.5){\makebox(-4,-2){$\cdots$}}
\put(15.5,3.5){\vector(-3,-1){8}\makebox(8,-2){$a$}}

\qbezier(15.5,3.5) (9,-0.5) (7.5,0.8)
\put(7.5,0.8){\vector(-1,-1){0.1}\makebox(7,1.6){b}}

\end{picture}
  \centering\caption{DFA $\mathcal{A}(m)$ solving $A^{eq}(m)$}\label{f1}
\end{figure}

\begin{Th}\label{DFAforAmeq}
For any $m\in {\mathbb{Z}}^+$, any DFA solving the promise problem
$A^{eq}(m)$ has at least $2m+2$ states.
\end{Th}
\begin{proof}
Let us consider the string set $W=\{a^0, a^1,\cdots,a^m, a^mb^1,
a^mb^2,\cdots, a^mb^m\}$, where $a^0$ is the empty string.
Obviously, for any two different strings $w_i, w_j\in W$, we have
$|w_i|\neq |w_j|$, and if $|w_i|<|w_j|$, then $w_i$ is a prefix of
$w_j$. For any string $x\in \Sigma^*$ and any $\sigma\in \Sigma$,
let $\widehat{\delta}(s,\sigma
x)=\widehat{\delta}(\delta(s,\sigma),x)$; if $|x|=0$,
$\widehat{\delta}(s,x)=s$ \cite{Hop79}. Assume that a $n$-state
DFA $\mathcal{A}(m)$ solves promise problem $A^{eq}(m)$. We show
that $n$ cannot be less than $2m+2$.

Assume that $s_0$ is the initial state of $\mathcal{A}(m)$, and
that there are two different strings $w_i, w_j\in W$ such that
$\widehat{\delta}(s_0,w_i)= \widehat{\delta}(s_0,w_j)$. Without a
lost of generality, we assume that $w_i$ is a prefix of $w_j$, so
there is a string $x$ such that $w_j=w_ix$, where $|x|\neq 0$. Let
$\widehat{\delta}(s_0,w_i)=s$, we have $\widehat{\delta}(s,
x)=\widehat{\delta}(s, x^*)=s$. Because $w_i$ is a prefix of
$a^mb^m$, there exists a string $y$ satisfies that
$\widehat{\delta}(s_0, w_iy)=\widehat{\delta}(s, y)=s_{acc}$,
where $s_{acc}$ is an accepting state. It follow
$\widehat{\delta}(s_0, w_ix^*y)=s_{acc}$. Therefore, there is some
$k\in {\mathbb{Z}}^+$ satisfy that $\widehat{\delta}(s_0,
w_ix^ky)=s_{acc}$ and $w_ix^ky\in A^{eq}_{no}(m)$, which is a
contradiction. Hence, for any two different strings $w_i, w_j\in
W$ satisfy that $\widehat{\delta}(s_0,w_i)\neq
\widehat{\delta}(s_0,w_j)$.

For any $w_i\in W$,  $\widehat{\delta}(s_0, w_i)$ is a reachable
state (i.e., there exists a string $z$ such that
$\widehat{\delta}(\widehat{\delta}(s_0, w_i),z)$ is an accepting
state). Therefore, there must be at least one state that is not
          reachable, for example, $\widehat{\delta}(s_0, a^mb^{m+1})$. There is $2m+1$ elements in the set $W$ and at least one not reachable state. So any DFA solving the promise problem $A^{eq}(m)$ has at least $2m+2$ states.

\end{proof}

\begin{Th}\label{S-result-2}
For any $m\in {\mathbb{Z}}^+$, any 2DFA, 2NFA and any polynomial
expected running time 2PFA solving the promise problem $A^{eq}(m)$
has at least $\sqrt{\log{m}}$, $\sqrt{\log{m}}$ and
$\sqrt[3]{(\log m)/b}$ states, where $b$ is a constant.
\end{Th}
\begin{proof}

Assume that an $n_1$-state 2DFA $\mathcal{A}$ solves the promise
problem $A^{eq}(m)$. It is easy to prove that $n_1\geq 3$.
According to Lemma \ref{2DFAtoDFA}, there is a DFA that solves the
promise problem $A^{eq}(m)$ with $(n_1+1)^{n_1+1}$ states.
According to Theorem \ref{DFAforAmeq}, we have
\begin{equation}
(n_1+1)^{n_1+1}\geq 2m+2 \Rightarrow (n_1+1)\log{(n_1+1)}>
\log{m}+1.
\end{equation}
Because $n_1\geq 3$, we get
\begin{equation}
n_1^2>(n_1+1)\log{(n_1+1)}> \log{m} \Rightarrow n> \sqrt{\log{m}}.
\end{equation}

Assume that an $n_2$-state 2NFA $\mathcal{A}$ solves the the
promise problem $A^{eq}(m)$.  According to Lemma \ref{2NFAtoDFA},
there is a DFA that solves the promise problem $A^{eq}(m)$ with
$2^{(n_2-1)^2+n_2}$ states. According to Theorem \ref{DFAforAmeq},
we have
\begin{equation}
2^{(n_2-1)^2+n_2}\geq 2m+2 \Rightarrow (n_2-1)^2+n_2> \log{m}+1
\end{equation}
\begin{equation}
\Rightarrow n_2^2> \log{m} \Rightarrow n_2> \sqrt{\log{m}}.
\end{equation}

Assume that an $n_3$-state 2PFA $\mathcal{A}$ solves the promise
problem $A^{eq}(m)$ with the error probability $\epsilon<1/2$ and
within a polynomial expected running time. According to Lemma
\ref{2PFAtoDFA}, there is a DFA that solves the promise problem
$A^{eq}(m)$ with $n_3^{bn_3^2}$ states, where $b>0$ is a constant.
According to Theorem \ref{DFAforAmeq}, we have
\begin{equation}
n_3^{bn_3^2}\geq 2m+2 \Rightarrow bn_3^2\log{n_3}> \log{m}
\end{equation}
\begin{equation}
\Rightarrow n_3^3> (\log{m})/ b\Rightarrow n_3>\sqrt[3]{(\log
m)/b}.
\end{equation}

\end{proof}

\section{State succinctness of 2QCFA}

For the alphabet $\Sigma=\{a,b,c\}$ and any $m\in {\mathbb{Z}}^+$,
let $L^{twin}(m)=\{wcw| w\in\{a,b\}^*,|w|=m\}$. For any
$\epsilon<1/2$, we will prove that $L^{twin}(m)$ can be recognized
by a 2QCFA  with one-sided error $\epsilon$ in an exponential
expected running time with a constant number of quantum states and
$\mathbf{O}(\log{\frac{1}{\epsilon})}$ classical states. The
language $L^{twin}=\{wcw | w\in\{a,b\}^*\}$ over alphabet
$\Sigma=\{a,b,c\}$ was declared as being recognized by a 2QCFA by
Yakaryilmaz and Cem Say \cite{Yak10}. However, they did not give
details of such a 2QCFA. In the following, we will show such an
automaton and its behavior in details.

\begin{Th}\label{Ltwin}
For any $\epsilon<1/2$, there exists a 2QCFA
$\mathcal{A}(\epsilon)$ which accepts any $w\in L^{twin}$ with
certainty,  rejects any $w\notin L^{twin}$ with probability at
least $1-\epsilon$, and halts in exponential expected time,  where
$QS(\mathcal{A}(\epsilon))$=3 and $CS(\mathcal{A}(\epsilon))\in
\mathbf{O}(\log{\frac{1}{\epsilon})}$.
\end{Th}
\begin{proof}
Let us consider $3\times3$ matrixes $U_{a}$ and $U_{b}$ defined as
follows:

\begin{equation}\label{matrix}
A=\left(
  \begin{array}{ccc}
      4 & 3 & 0  \\
     -3 & 4 & 0  \\
      0 & 0 & 5
  \end{array}
\right), B=\left(
  \begin{array}{ccc}
      4 & 0 & 3 \\
      0 & 5 & 0  \\
      -3& 0 & 4 \\
  \end{array}
\right).
\end{equation}
 We now describe formally a 2QCFA $\mathcal{A}(\epsilon)$ that is described less formally in Figure \ref{t2} with 3 quantum states $\{|q_0\rangle, |q_1\rangle, |q_2\rangle\}$, with $|q_0\rangle$ being the initial state.  $\mathcal{A}(\epsilon)$  has two unitary operators $U_{a}=\frac{1}{5}A$ and $U_{b}=\frac{1}{5}B$ given in Eq. (\ref{matrix}). They can also be described as follows: \\

\begin{tabular}{|l|l|l|}
  \hline
  $U_{a}|q_0\rangle=\frac{4}{5}|q_0\rangle-\frac{3}{5}|q_1\rangle$ & $U_{b}|q_0\rangle=\frac{4}{5}|q_0\rangle-\frac{3}{5}|q_2\rangle$  \\
  $U_{a}|q_1\rangle=\frac{3}{5}|q_0\rangle+\frac{4}{5}|q_1\rangle$ & $U_{b}|q_1\rangle=|q_1\rangle$                                    \\
  $U_{a}|q_2\rangle=|q_2\rangle$                                   & $U_{b}|q_2\rangle=\frac{3}{5}|q_0\rangle+\frac{4}{5}|q_2\rangle$  \\

  \hline
\end{tabular}\\

\begin{figure}
  \centering
  \begin{tabular}{|l|}
  \hline
Check whether the input is of the form $xcy$ ($x,y\in\{a,b\}^*$). If not, reject.\\
Otherwise, repeat the following ad infinitum:\\
\ \ 1. Move the tape head to the first input symbol and set the quantum state to $|q_0\rangle$.\\
\ \ 2. Until the currently scanned symbol $\sigma$ is $c$, do the following:\\
\ \ \ \ (2.1).Perform $U_{\sigma}$ on the quantum state.\\
\ \ \ \ (2.2).Move the tape head one square to the right.\\
\ \ 3. Move the tape head to the last input symbol.\\
\ \ 4. Until the currently scanned symbol $\sigma$ is $c$, do the following:\\
\ \ \ \ (4.1).Perform $U_{\sigma}^{-1}$ on the quantum state.\\
\ \ \ \ (4.2).Move the tape head one square to the left.\\
\ \ 5. Measure the quantum state. If the result is not $|q_0\rangle$, reject.\\
\ \ 6.  Move the tape head to the last input symbol and set $b=0$.\\
\ \ 7. While the currently scanned symbol is not $\ |\hspace{-1.5mm}c$, do the following:\\
\ \ \ \ (7.1). Simulate $k$ coin-flips. Set $b=1$ in case all results are not ``heads".\\
\ \ \ \ (7.2). Move the tape head one square to the left. \\
\ \ 8. If $b=0$, accept.\\

  \hline
\end{tabular}\\
  \centering\caption{Informal description of the actions of a 2QCFA
                   for $L^{twin}$. The choice of $k$ will depend on $\epsilon$.}\label{t2}
\end{figure}

We now summarize some concepts and results from \cite{Amb02} that
we will use to prove the theorem. For $u\in \mathbb{Z}^3$, we use
$u[i]\ (i=1,2,3)$ to denote the $i$th entry of $u$. We define a
function $f:\mathbb{Z}^3\rightarrow\mathbb{Z}$ as
\begin{equation}
 f(u)=4u[1]+3u[2]+3u[3]
\end{equation}
for each $u\in \mathbb{Z}^3$, and we define a set
$K\subseteq\mathbb{Z}^3$ as
\begin{equation}
K=\{u\in \mathbb{Z}^3: u[1]\not\equiv 0(mod \ 5),
f(u)\not\equiv0(mod\ 5), {\it and } \ u[2]\cdot u[3]\equiv0(mod\
5)\}
\end{equation}

\begin{Lm}[\cite{Amb02}]\label{Au}
If $u\in K$, then $Au\in K$ and $Bu\in K$.
\end{Lm}

\begin{Lm}[\cite{Amb02}]\label{u-notin-K}
If an $u\in \mathbb{Z}^3$ is such that $u=Av=Bw$ for some $v,w\in
\mathbb{Z}^3$, then $u\notin K$.
\end{Lm}

\begin{Lm} \label{no-eq-q0}
If $u\in K$, there does not exist an $l\in \mathbb{Z}^+$ such that
$Xu=\pm 5^l(1, 0 ,0)^T$, where $X\in \{A, B\}$.
\end{Lm}
\begin{proof}
Suppose there is an $l\in \mathbb{Z}^+$ such that $Xu=\pm 5^l(1, 0
,0)^T$. Assume that $X=A$ (the proof for $X=B$ is similar), then
it holds
\begin{equation}
 Xu=Au=\left(
  \begin{array}{ccc}
      4 & 3 & 0 \\
      -3 & 4 & 0  \\
      0& 0 & 5 \\
  \end{array}
\right)\left(
  \begin{array}{c}
      u[1]  \\
      u[2]  \\
      u[3] \\
  \end{array}
\right)=\left(
  \begin{array}{c}
      4u[1]+3u[2]  \\
      -3u[1]+4u[2]  \\
      5u[3] \\
  \end{array}
\right)=\pm\left(
  \begin{array}{c}
      1  \\
      0  \\
      0 \\
  \end{array}
\right)5^l
\end{equation}

\begin{equation}
\Rightarrow \left(
  \begin{array}{c}
      u[1]  \\
      u[2]  \\
      u[3] \\
  \end{array}
\right)=\pm\left( \begin{array}{c}
      4\cdot 5^{l-2} \\
      3\cdot 5^{l-2} \\
      0 \\
  \end{array} \right).
\end{equation}
Since $4u[1]+3u[2]+3u[3]=\pm(16\cdot 5^{l-2}+9\cdot 5^{l-2})=\pm
5^{l}$, we conclude $f(u)\equiv0(mod\ 5)$. We get that $u\notin
K$, which contradicts the fact that $u\in K$. Hence, the Lemma has
been proved.
\end{proof}

\begin{Co}\label{Co-u-neq-q0}
Let
\begin{equation}
u=X_k\cdots X_1(1, 0, 0)^T,
\end{equation}
where $X_i\in \{A, B\}$. Then $u=\pm 5^l(1, 0, 0)^T$ for
         no $l\in \mathbb{Z}^+$.
\end{Co}
\begin{proof}
Clearly, $(1, 0, 0)^T\in K$. According to Lemma \ref{Au},
$X_{k-1}\cdots X_1(1, 0, 0)^T\in K$. According to Lemma
\ref{no-eq-q0}, there does not exist $l\in \mathbb{Z}^+$ such that
$u=\pm 5^l(1, 0, 0)^T$.
\end{proof}

\begin{Lm}\label{m>n}
Let
\begin{equation}
u=Y_1^{-1}\cdots Y_k^{-1}(1, 0, 0)^T,
\end{equation}
where $Y_i\in \{A, B\}$. Then $u=\pm \frac{1}{5^l}(1, 0, 0)^T$ for
         no $l\in \mathbb{Z}^+$.
\end{Lm}
\begin{proof}
Assume that there is an $l\in \mathbb{Z}^+$ satisfies that
$u=Y_1^{-1}\cdots Y_k^{-1}(1, 0, 0)^T=\pm \frac{1}{5^l}(1, 0,
0)^T$, then we get $Y_k\cdots Y_1(1, 0, 0)^T=\pm 5^l (1, 0, 0)^T$.
According to Corollary \ref{Co-u-neq-q0}, such $l$ does not exist.
\end{proof}

\begin{Lm}\label{XY}
Let
\begin{equation}
u=(5Y_1^{-1})\cdots (5Y_m^{-1})(5^{-1}X_n)\cdots
(5^{-1}X_1)(1,0,0)^T,
\end{equation}
where $X_j, Y_j\in \{A, B\}$. If $m=n$ and $X_j= Y_j$ for $1\leq
j\leq n$, then $u[2]^2+u[3]^2=0$. Otherwise,
$u[2]^2+u[3]^2>5^{-(n+m)}$.
\end{Lm}
\begin{proof}
If $m=n$ and $X_j= Y_j$ for $1\leq j\leq n$, then we have
\begin{equation}
 u=Y_1^{-1}\cdots Y_n^{-1}X_n\cdots X_1(1,0,0)^T=(1,0,0)^T,
\end{equation}
and thus $u[2]^2+u[3]^2=0$.

Otherwise, note that $||u||=1$, since $5^{-1}X_j$ and $5Y_j^{-1}$
are unitary for each $j$, and also note that $5^{(n+m)}u[i]\
(i=1,2,3)$ is an integer. It therefore suffices to prove that
$u\neq \pm (1,0,0)^T$. $|u[1]|<1$ implies $|u[1]|\leq
1-5^{-(n+m)}$, and therefore
\begin{equation}
u[2]^2+u[3]^2=1-u[1]^2\geq 1-(1-5^{-(n+m)})^2>5^{-(n+m)}.
\end{equation}
 We first prove the case that $n\geq m$. If $X_{n-j}=Y_{m-j}$ for $0\leq j\leq m-1$,
then
\begin{equation}
 u=(5Y_1^{-1})\cdots (5Y_m^{-1})(5^{-1}X_n)\cdots (5^{-1}X_1)(1,0,0)^T=5^{-(n-m)}X_{n-m}\cdots X_1(1,0,0)^T.
\end{equation}
According to Corollary \ref{Co-u-neq-q0}, for every $l\in
\mathbb{Z}^+$,
\begin{equation}
u=5^{-(n-m)}X_{n-m}\cdots X_1(1,0,0)^T\neq \pm 5^{-(n-m)}5^l
(1,0,0)^T.
\end{equation}
This implies that $u\neq\pm (1,0,0)^T$ if $l=n-m$.

Next suppose there exist an $i<m$ such that $X_{n-i}\neq Y_{m-i}$.
Let $k$ be the smallest integer such that $X_{n-k}\neq Y_{m-k}$,
and without loss of generality suppose $X_{n-k}=A, Y_{m-k}=B$.
Since $X_{n-j}=Y_{m-j}$ for $j<k$, we have
\begin{equation}
 u=(5Y_1^{-1})\cdots (5Y_m^{-1})(5^{-1}X_n)\cdots (5^{-1}X_1)(1,0,0)^T=5^{-(n-m)}Y^{-1}_1\cdots Y^{-1}_{m-k}X_{n-k}\cdots X_1(1,0,0)^T.
\end{equation}
For $u=(1,0,0)^T$, we get
\begin{equation}
u=5^{-(n-m)}Y^{-1}_1\cdots Y^{-1}_{m-k}X_{n-k}\cdots
X_1(1,0,0)^T=(1,0,0)^T
\end{equation}
\begin{equation} \label{e1}
\Rightarrow X_{n-k}\cdots X_1(1,0,0)^T=5^{(n-m)}Y_{m-k}\cdots
Y_{1}(1,0,0)^T=Y_{m-k}\cdots Y_{1} 5^{(n-m)} (1,0,0)^T
\end{equation}

Obviously, $(1,0,0)^T\in K$ and $5^{(n-m)}(1,0,0)^T\in K$. Let
$v=X_{n-k-1}\cdots X_1(1, 0 ,0)^T$ and $w=Y_{m-k-1}\linebreak[0]
\cdots\linebreak[0]
Y_1\linebreak[0]5^{(n-m)}\linebreak[0](1,0,0)^T$, according to
Lemma \ref{Au}, we have $v,w\in K$, $X_{n-k}v=Av\in K$, and
$Y_{m-k}w=Bw\in K$. By Lemma \ref{u-notin-K} this implies $Av\neq
Bw$, which contradicts the Equation \ref{e1}. From that we
conclude $u\neq (1,0,0)^T$. By similar reasoning we get that,
$u\neq -(1,0,0)^T$.

Now we deal with the case $n<m$. If $X_{n-j}=Y_{m-j}$ for $0\leq
j\leq n-1$, then
\begin{equation}
 u=(5Y_1^{-1})\cdots (5Y_m^{-1})(5^{-1}X_n)\cdots (5^{-1}X_1)(1,0,0)^T=5^{m-n}Y_1^{-1}\cdots Y_{m-n}^{-1}(1,0,0)^T.
\end{equation}
According to Lemma \ref{m>n}, for every $l\in \mathbb{Z}^+$,
\begin{equation}
u=5^{m-n}Y_1^{-1}\cdots Y_{m-n}^{-1}(1,0,0)^T\neq \pm 5^{m-n}
5^{-l} (1,0,0)^T.
\end{equation}
This implies that $u\neq \pm(1,0,0)^T$ if $l=m-n$.

Let us assume that there exist $j<n$ such that $X_{n-j}\neq
Y_{m-j}$. Let $k$ be the smallest index such that $X_{n-k}\neq
Y_{m-k}$. By similar reasoning as in the case $n\geq m$, we get
$u\neq \pm(1,0,0)^T$.
\end{proof}

If the input $w$ is not of the form $xcy$, $\mathcal{A}(\epsilon)$
rejects $w$ immediately.
\begin{Lm}
If the input $w=xcy$ and $x=y$, then the quantum state of
$\mathcal{A}(\epsilon)$ will evolve into $|q_0\rangle$ with
certainty after the loop \textbf{4}.
\end{Lm}

\begin{proof}
Let $x=x_1x_2\ldots x_l=y=y_1y_2\ldots y_l$ for some $l$. Starting
with the state $|q_0\rangle$, $\mathcal{A}(\epsilon)$ changes its
quantum state to  $|\psi\rangle$ after the loop \textbf{4}, where
\begin{equation}
 |\psi\rangle=U_{y_1}^{-1}U_{y_2}^{-1}\cdots U_{y_l}^{-1}U_{x_l}\cdots U_{x_2}U_{x_1}|q_0\rangle=U_{x_1}^{-1}U_{x_2}^{-1}\cdots U_{x_l}^{-1}U_{x_l}\cdots U_{x_2}U_{x_1}|q_0\rangle=|q_0\rangle.
\end{equation}
\end{proof}

\begin{Lm}\label{x-neq-y}
If the input $w=xcy$ and $x\neq y$, then $\mathcal{A}(\epsilon)$
rejects $w$ after the step \textbf{5} with the probability at
least $5^{-(m+n)}$.
\end{Lm}
\begin{proof}
Let $x=x_1x_2\cdots x_n$, $y=y_1y_2\cdots y_m$. Starting with
state $|q_0\rangle$, $\mathcal{A}(\epsilon)$ changes its quantum
state after the loop 4 to:
\begin{equation}
 |\psi\rangle=U_{y_1}^{-1}U_{y_2}^{-1}\cdots U_{y_m}^{-1}U_{x_n}\cdots U_{x_2}U_{x_1}|q_0\rangle.
\end{equation}
Let
$|\psi\rangle=\beta_0|q_0\rangle+\beta_1|q_1\rangle+\beta_2|q_2\rangle$.
According to Lemma \ref{XY}, $\beta_1^2+\beta_2^2>5^{-(n+m)}$. In
the step \textbf{5}, the quantum state $|\psi\rangle$ is measured,
$\mathcal{A}(\epsilon)$ then rejects $w$ with the probability
$p_r=\beta_1^2+\beta_2^2>5^{-(n+m)}$.
\end{proof}

Every execution of the steps \textbf{6}, \textbf{7} and \textbf{8}
leads to an acceptance with the probability $2^{-k(n+m+1)}$.

Let $k\geq \max\{\log{5}, \log{\frac{1}{\epsilon}}\}$. Assume that
the input is of the form $w=xcy$. If $x=y$, 2QCFA
$\mathcal{A}(\epsilon)$ always changes its quantum state to
$|q_0\rangle$ after the loop \textbf{4}, and
$\mathcal{A}(\epsilon)$ never rejects the input after the
measurement in the step $\mathbf{5}$. After the steps \textbf{6},
\textbf{7} and \textbf{8}, the probability of
$\mathcal{A}(\epsilon)$ accepting $w$ is $2^{-k(n+m+1)}$.
Repeating the whole iteration for $c2^{k(n+m+1)}$ times, the
accepting probability is
\begin{equation}
Pr[\mathcal{A}(\epsilon)\  {\it accepts}\  w]
=1-(1-2^{-k(n+m+1)})^{c2^{k(n+m+1)}},
\end{equation}
and this can be made arbitrarily close to 1 by selecting constant
$c$ appropriately.

Otherwise, if $x\neq y$, then, according to Lemma \ref{x-neq-y},
$\mathcal{A}(\epsilon)$ rejects the input after the step
$\mathbf{5}$ with the probability
\begin{equation}
P_{r} >5^{-(m+n)}
\end{equation}
and,  $\mathcal{A}(\epsilon)$ accepts the input after the steps
\textbf{6}, \textbf{7} and \textbf{8} with the probability

\begin{equation}
P_{a}=2^{-k(n+m+1)}.
\end{equation}
If we repeat the whole iteration indefinitely, the probability of
$\mathcal{A}(\epsilon)$ rejecting input $w$ is
\begin{equation}
Pr[\mathcal{A}(\epsilon)\  {\it rejects}\  w] =\sum_{i\geq
0}(1-P_a)^i(1-P_r)^iP_r
\end{equation}
\begin{equation}
=\frac{P_r}{P_a+P_r-P_aP_r}>\frac{P_r}{P_a+P_r}
\end{equation}
\begin{equation}
>\frac{5^{-(m+n)}}{2^{-k(n+m+1)}+5^{-(m+n)}}
\end{equation}
\begin{equation}
>\frac{1}{1+\epsilon}>1-\epsilon.
\end{equation}

If the input is $w$, then the step $\mathbf{1}$ takes
$\mathbf{O}(1)$ time, the steps $\mathbf{2}$ and $\mathbf{3}$ take
$\mathbf{O}(|w|)$ time, the loops $\mathbf{4}$ and $\mathbf{5}$
take $\mathbf{O}(|w|)$ time, the steps \textbf{6}, \textbf{7} and
\textbf{8} take $\mathbf{O}(|w|)$ time.  The expected number of
iterations
 is $\mathbf{O}(2^{k|w|})$. Hence, the expected running
time of $\mathcal{A}(\epsilon)$ is $\mathbf{O}(|w|2^{k|w|})$.
Obviously, the 2QCFA $\mathcal{A}(\epsilon)$ has three quantum
states. We just need $\mathbf{O}(k)$ classical states to simulate
$k$ coin-flips and calculate the outcomes, therefore
$CS(\mathcal{A}(\epsilon))\in
\mathbf{O}(\log{\frac{1}{\epsilon})}$.

\end{proof}

In Theorem \ref{Ltwin}, we have proved that $L^{twin}$ can be
recognized by 2QCFA. We will show that $L^{twin}$ can not be
recognized by 2PFA with error probability $\epsilon<1/2$. Thus
$L^{twin}$ is another witness of the fact that 2QCFA are more
powerful than their classical counterparts 2PFA.

\begin{Th}
There is no 2PFA recognizing $L^{twin}$ with error probability
$\epsilon<1/2$.
\end{Th}
\begin{proof}
Let $A=L^{twin}$ and $B=\overline{L^{twin}}=\Sigma^*\setminus A$.
Clearly, for each $m\in I$, there is a set $W_m\subseteq\Sigma^*$
satisfying conditions $(1)$ and $(2)$ of Lemma \ref{not-in-2PFA}.
For every $m\in I$ and every $w, w'\in W_m$ with $w\neq w'$, if we
take $u=\lambda$ (the empty word) and $v=cw$, then $uwv=wcw\in A$
and $uw'v=w'cw\in B$. According to Lemma \ref{not-in-2PFA}, there
is no 2PFA separating $A$ and $B$. Thus, there is no 2PFA
recognizing $L^{twin}$ and the Theorem has been proved.
\end{proof}

For an alphabet $\Sigma$ and an $m\in {\mathbb{Z}}^+$, let
$L(m)=\{w\mid |w|=m\}$.
\begin{Lm}[\cite{Yak10}]\label{Lm}
For any $\epsilon<1/2$, there exists a $7$-state
1QFA$^\circlearrowleft$ $\mathcal{A}(m,\epsilon)$ which accepts
any $w\in L(m)$ with certainty, and rejects any $w\notin L(m)$
with probability at least $1-\epsilon$. Moreover, the expected
runtime of the $\mathcal{A}(m,\epsilon)$ on $w$ is
$\mathbf{O}(2^{|w|}|w|)$.
\end{Lm}

\begin{Lm}[\cite{Yak10}]\label{1QFA-2QCFA}
For any 1QFA$^\circlearrowleft$ $\mathcal{A}_1$ with $n$ quantum
states and expected running time $t(|w|)$, there exists a 2QCFA
$\mathcal{A}_2$ with $n$ quantum states, $\mathbf{O}(n)$ classical
states, and expected running time $\mathbf{O}(t(|w|))$, such that
$\mathcal{A}_2$ accepts every input string $w$ with the same
probability that $\mathcal{A}_1$ accepts $w$.
\end{Lm}
\begin{Th} \label{Th-lm-2QCFA}
For any $\epsilon<1/2$, there exists a 2QFA
$\mathcal{A}(m,\epsilon)$ which accepts any $w\in L(m)$ with
certainty, and rejects any $w\notin L(m)$ with probability at
least $1-\epsilon$. Moreover, $QS(\mathcal{A}(m,\epsilon))=7$,
$CS(\mathcal{A}(m,\epsilon))$ is a constant, and the expected
runtime of the $\mathcal{A}(m,\epsilon)$ on $w$ is
$\mathbf{O}(2^{|w|}|w|)$.
\end{Th}
\begin{proof}
It follows from Lemma \ref{Lm} and Lemma \ref{1QFA-2QCFA}.
\end{proof}

\begin{Th}
For any $\epsilon<1/2$, there exists a 2QFA
$\mathcal{A}(m,\epsilon)$ which accepts any $w\in L^{twin}(m)$
with certainty, and rejects any $w\notin L^{twin}(m)$ with
probability at least $1-\epsilon$. Moreover,
$QS(\mathcal{A}(m,\epsilon))$ is a constant,
$CS(\mathcal{A}(m,\epsilon))\in
\mathbf{O}(\log{\frac{1}{\epsilon})}$, and the expected running
time of $\mathcal{A}(m,\epsilon)$ on $w$ is
$\mathbf{O}(|w|2^{k|w|})$.
\end{Th}
\begin{proof}
Let the alphabet $\Sigma=\{a,b,c\}$. Obviously,
$L^{twin}(m)=L^{twin}\cap L(2m+1)$. According to Theorem
\ref{Ltwin}, for any $\epsilon_1<1/2$, there is a 2QCFA
$\mathcal{A}_1(\epsilon_1)$ recognizes $L^{twin}$ with one-sided
error $\epsilon_1$, and $QS(\mathcal{A}_1(\epsilon_1))=3$,
$CS(\mathcal{A}_1(\epsilon_1))\in\mathbf{O}(\log{\frac{1}{\epsilon_1})}$
and $T(\mathcal{A}_1(\epsilon_1)\in\mathbf{O}(|w|2^{k|w|})$ where
$k$ is a constant. According to Theorem \ref{Th-lm-2QCFA}, for any
$\epsilon_2<1/2$, there is a 2QCFA $\mathcal{A}_2(m,\epsilon)$
recognizes $L(2m+1)$ with one-sided error $\epsilon_2$, and
$QS(\mathcal{A}_2(m,\epsilon))=7$, $CS(\mathcal{A}_2(m,\epsilon))$
is a constant and
$T(\mathcal{A}_2(m,\epsilon))\in\mathbf{O}(2^{|w|}|w|)$. For any
$\epsilon<1/2$, let $\epsilon_1=\epsilon/2$ and
$\epsilon_2=\epsilon/2$. According to Lemma \ref{Lm_2qcfa_inter},
there is a 2QCFA $\mathcal{A}(m,\epsilon)$ recognizes
$L^{twin}\cap L(2m+1)$ with one-sided error $\epsilon$, where
$QS(\mathcal{A}(m,\epsilon))=QS(\mathcal{A}_1(\epsilon_1))+
QS(\mathcal{A}_2(m,\epsilon))=10$,
$CS(\mathcal{A}(m,\epsilon))=CS(\mathcal{A}_1(\epsilon_1))+CS(\mathcal{A}_2(m,\epsilon))+QS(\mathcal{A}_1(\epsilon_1))\in
\mathbf{O}(\log{\frac{1}{\epsilon})}$ and
$T(\mathcal{A}(m,\epsilon))=T(\mathcal{A}_1(\epsilon_1))+T(\mathcal{A}_2(m,\epsilon))\in\mathbf{O}(|w|2^{k|w|})$.
Hence, the theorem has been proved.
\end{proof}

For a fix $m\in {\mathbb{Z}}^+$, $L^{twin}(m)$ is finite, and thus
there exists a DFA accepting the language $L^{twin}(m)$. In the
following we use methods and results of {\it communication
complexity} to derive a lower bound on the number of states of
finite automata accepting the language $L^{twin}(m)$

\begin{Th}\label{DFAforAmtwin}
For any $m\in {\mathbb{Z}}^+$, any DFA recognizing $L^{twin}(m)$
has at least $2^m$ states.
\end{Th}
\begin{proof}
Assume that a DFA $\mathcal{A}$ recognizes $L^{twin}(m)$.
 For an input string $xcy$ of
$L^{twin}(m)$ let us consider the following communication protocol
between Alice and Bob with Alice having $x$ at the beginning and
Bob having $y$ at the beginning. A protocol can be derived for
$EQ(x,y)$ as follows:
  Alice first simulates the path taken by DFA $\mathcal{A}$ on her input $x$.  She then sends the name of the last state $s$ in this path to Bob, which needs $\log{(|S|)}$ bits, where $S$ is the set of states in DFA $\mathcal{A}$. Afterwards, Bob simulates DFA $\mathcal{A}$, starting from the state $s$, on input $cy$. At last, Bob sends the result to Alice, if $w$ is accepted, bob sends 1, otherwise 0.  All together, they get a simulation of DFA  $\mathcal{A}$ on the input $w=xcy$. By assumption, if $w=xcy$ is accepted by DFA $\mathcal{A}$ then $EQ(x,y)=1$ while if $w$ is rejected then $EQ(x,y)=0$. Therefore, we have $D(EQ)\leq \log{(|S|)}+1$. According to Lemma \ref{ComC-EQ}, we have
\begin{equation}
    D(EQ)=m+1\leq \log{(|S|)}+1
\end{equation}
\begin{equation}
   \Rightarrow m\leq \log{(|S|)} \Rightarrow |S|\geq 2^{m}.
\end{equation}
\end{proof}

\begin{Th}\label{S-result-2}
For any $m\in {\mathbb{Z}}^+$, any 2DFA, 2NFA and polynomial
expected running time 2PFA recognizing $L^{twin}(m)$ have at least
$\sqrt{m}$, $\sqrt{m}$ and $\sqrt[3]{m/b}$ states, where $b$ is a
constant.
\end{Th}
\begin{proof}
Assume that an $n_1$-state 2DFA $\mathcal{A}$ recognizes
$L^{twin}(m)$. It is easy to prove that $n_1\geq 3$. According to
Lemma \ref{2DFAtoDFA}, there is a DFA recognizes $L^{twin}(m)$
with $(n_1+1)^{n_1+1}$ states. According to Theorem
\ref{DFAforAmtwin}, we have
\begin{equation}
(n_1+1)^{n_1+1}\geq 2^m \Rightarrow (n_1+1)\log{(n_1+1)}\geq m.
\end{equation}
Because $n\geq 3$, we get
\begin{equation}
n_1^2>(n_1+1)\log{(n_1+1)}> m \Rightarrow n_1> \sqrt{m}.
\end{equation}

Assume that an $n_2$-state 2NFA $\mathcal{A}$ recognizes
$L^{twin}(m)$.  According to Lemma \ref{2NFAtoDFA}, there is a DFA
recognizes $L^{twin}(m)$ with $2^{(n_2-1)^2+n_2}$ states.
According to Theorem \ref{DFAforAmtwin}, we have
\begin{equation}
2^{(n_2-1)^2+n_2}\geq 2^m \Rightarrow (n_2-1)^2+n_2\geq m
\end{equation}
\begin{equation}
\Rightarrow n_2^2> m \Rightarrow n_2> \sqrt{m}.
\end{equation}

Assume that an $n_3$-state 2PFA $\mathcal{A}$ recognizes
$L^{twin}(m)$ with an error probability $\epsilon<1/2$ and within
a polynomial expected running time. According to Lemma
\ref{2PFAtoDFA}, there is a DFA recognizes $L^{twin}(m)$ with
$n_3^{bn_3^2}$ states, where $b>0$ is a constant. According to
Theorem \ref{DFAforAmtwin}, we have
\begin{equation}
n_3^{bn_3^2}\geq 2^m \Rightarrow bn_3^2\log{n_3}\geq m
\end{equation}
\begin{equation}
\Rightarrow n_3^3> m/ b\Rightarrow n_3>\sqrt[3]{m/b}.
\end{equation}

\end{proof}

\section{Concluding remarks}

2QCFA were introduced by Ambainis and Watrous \cite{Amb02}. In
this paper, we investigated state succinctness of 2QCFA. We have
showed that 2QCFA can be more space-efficient than their classical
counterparts DFA, 2DFA, 2NFA and polynomial expected running time
2PFA, where the superiority cannot be bounded. For any $m\in
{\mathbb{Z}}^+$ and any $\epsilon<1/2$, we have proved that there
is a promise problem $A^{eq}(m)$ that can be solved by a 2QCFA
with one-sided error $\epsilon$ in a polynomial expected running
time with a constant number of quantum states and
$\mathbf{O}(\log{\frac{1}{\epsilon})}$ classical states, whereas
the sizes of the corresponding DFA, 2DFA, 2NFA and polynomial
expected running time 2PFA are at least $2m+2$, $\sqrt{\log{m}}$,
$\sqrt{\log{m}}$ and $\sqrt[3]{(\log m)/b}$. For any $m\in
{\mathbb{Z}}^+$ and any $\epsilon<1/2$, we have also showed that
there exists a 2QCFA recognizing the language $L^{twin}(m)$ with
one-sided error $\epsilon$ in an exponential expected running time
with a constant number of quantum states and
$\mathbf{O}(\log{\frac{1}{\epsilon})}$ classical states, whereas
the sizes of the corresponding DFA, 2DFA, 2NFA and polynomial
expected running time 2PFA are at least $2^m$, $\sqrt{m}$,
$\sqrt{m}$ and $\sqrt[3]{m/b}$.

To conclude, we formulate some open problems:
\begin{enumerate}
 \item Can the result related to a promise problem $A^{eq}(m)$ be improved to deal with a language?
  \item In Theorem \ref{S-result-2}, we gave a bound on a polynomial expected running time 2PFA. What is the bound when the expected running time is exponential?
\end{enumerate}

\section*{Acknowledgements}

This work is supported in part by the National Natural Science
Foundation (Nos. 60873055, 61073054, 61100001), the Natural
Science Foundation of Guangdong Province of China (No.
10251027501000004),  the Research Foundation for the Doctoral
Program of Higher School of Ministry of Education of China (Nos.
20100171110042, 20100171120051), the Fundamental Research Funds
for the Central Universities (Nos. 10lgzd12,11lgpy36), the China
Postdoctoral Science Foundation project (Nos. 20090460808,
201003375), and the project of SQIG at IT, funded by FCT and EU
FEDER projects QSec PTDC/EIA/67661/2006,  FCT project
PTDC/EEA-TEL/103402/2008 QuantPrivTel, FCT
PEst-OE/EEI/LA0008/2011, AMDSC UTAustin/MAT/0057/2008, IT Project
QuantTel, Network of Excellence, Euro-NF.

\end{document}